\documentclass[submission,copyright,creativecommons]{eptcs}
\providecommand{\event}{AFL 2023} 

\usepackage{iftex}

\ifpdf
  \usepackage{underscore}         
  \usepackage[T1]{fontenc}        
\else
  \usepackage{breakurl}           
\fi

\usepackage{amssymb}
\usepackage{amsthm}
\usepackage{latexsym}

\usepackage[all]{xy}
\usepackage[american]{babel}
\usepackage{amsmath}
\usepackage{textcomp}

\newcommand{\TE}{\mathit{TR}}

\newcommand{\mydot}{\mathrel{\scriptscriptstyle\diamond}}

\newcommand{\TC}{\mathit{TC}}
\newcommand{\TR}{\mathit{TR}}

\usepackage{tikz}
\usepackage{marvosym}
\usetikzlibrary{shapes}
\usetikzlibrary{arrows,arrows.meta}
\usetikzlibrary{calc,through,backgrounds}
\usetikzlibrary{shapes.geometric,fit}

        \newtheorem{theorem}{Theorem}
        \newtheorem{proposition}{Proposition}
        \newtheorem{lemma}{Lemma}

         \newtheorem{definition}{Definition}

\newcounter{example}
\newenvironment{example}[1][]{\refstepcounter{example}\par\medskip
   \noindent \textit{Example~\theexample. #1} \rmfamily}{\medskip}

\newcommand{\Conf}{\mathit{Conf}}

\title{Comparative Transition System Semantics for\\ Cause-Respecting Reversible Prime Event Structures}
\author{Nataliya Gribovskaya \qquad Irina Virbitskaite
\institute{A.P. Ershov Institute of Informatics Systems\\
the Siberian Branch of the Russian Academy of Sciences\\
6, Acad. Lavrentiev avenue,  630090, Novosibirsk, Russia}
\email{\{natamosk,virbitskaite\}@gmail.com}
}
\def\titlerunning{Comparative Transition System Semantics for Cause-Respecting RPESs}
\def\authorrunning{N. Gribovskaya, I. Virbitskaite}
\begin{document}
\maketitle

\begin{abstract}
Reversible computing is a new paradigm that has emerged recently and extends the traditional forwards-only computing mode with the ability to execute in backwards,
so that computation can run in reverse as easily as in forward.
Two approaches to developing transition system (automaton-like) semantics for event structure models are distinguished in the literature.
In the first case, states are considered as configurations (sets of already executed events),
and transitions between states are built by starting from the initial configuration and repeatedly adding executable events.
In the second approach, states are understood as residuals (model fragments that have not yet been executed),
and transitions are constructed by starting from the given event structure as the initial state and deleting already executed (and conflicting) parts thereof during execution.
The present paper focuses on an investigation of how the two approaches are interrelated for the model of prime event structures extended with cause-respecting reversibility.
The bisimilarity of the resulting transition systems is proved, taking into account step semantics of the model under consideration.
\end{abstract}
\section{Introduction}

Reversible computations, extensively studied during in recent years, is an unconventional form of computations that can be performed in the forward direction as easily as in the reverse direction.
Any sequence of actions executed by the system can subsequently be canceled for some reason (for example, in case of an error),
which allows the system to restore previous consistent states, as if these canceled actions were not executed at all.
Reversible computing is attracting interest for its applications in many fields including
program analysis and debugging \cite{LPV19},
programming abstractions for reliable systems \cite{DK05,MMPY20},
modelling biochemical reactions \cite{KACPPU20},
hardware design and quantum computing \cite{DDV18}, and etc.

Despite the fact that reversing computations in concurrent/distributed systems has many promising applications, it also involves many technical and conceptual challenges.
One of the most essential issues that arise concerns the techniques that should be applied when moving backwards.
Several different styles of the undoing of computation have been identified recently.
The most prominent of these are backtracking \cite{PP20}, causal reversibility \cite{MMPPU20,PP20}, and out-of-causal-order reversibility \cite{PP20,PUY13},
that differ in the order of executing actions in backward direction.
Backtracking is generally understood as the ability to execute past actions in the exact reverse order in which they were executed.
Causal reversibility in concurrent systems means that actions that cause others can only be undone after the caused actions are undone first,
and that actions which are independent of each other can be reversed in an arbitrary order.
Out-of-causal reversibility, a form of reversal most characteristic of biochemical systems, does not preserve causes.
The interplay between reversibility and concurrency has been widely studied in various models:
parallel rewriting systems \cite{AC18}, cellular automata \cite{K18}, process calculi \cite{DK05,LMS16}, Petri nets \cite{BGMPPP22,FKM19,PP20},
event structures \cite{MMP21,PU15,UPY18}, membrane systems \cite{P18}, and etc.

Event structures are a well-established model of concurrency.
They were originally proposed by Winskel in his PhD dissertation \cite{W80} and
were considered as an intermediate abstraction between Scott domains (i.e., a denotational model) and Petri nets (i.e., an operational model).
Basically, event structures are collections of possible events, some of which are conflicting (i.e., the execution of an event forbids the execution of other events),
while others are causally dependent (i.e., an event cannot be executed if it has not been preceded by other ones), and
events that are neither in causal dependency nor in conflict are treated as concurrent.
Events are often labelled with actions, to represent different occurrences of the same action.
Prime Event Structures (written PESs) are the earliest and simplest form of event structures,
where causality is a partial order and conflict between events is inherited by their causal successors.
The association of transition system (automaton-like) models with event structures has proved to contribute to studying and solving various problems in the analysis and verification of concurrent systems.
It is distinguished two methods of providing transition system semantics for event structures: a configuration-based and a residual-based method.
In the first case (see \cite{AKPN15,ABG16,GP09,HKT96,K96,W80,W89} among others),
states are understood as sets of events, called configurations,
and state transitions are built by starting with the initial configuration and enlarging configurations by already executed events.
In the second more `structural' method (see \cite{BM94,B90,CVY12,K96,L93} among others),
states are understood as event structures,
and transitions are built by starting with the given event structure as an initial state and removing already executed (and conflicting) parts thereof in the course of execution.
In the literature, configuration-based transition systems seem to be predominantly used as the semantics of event structures,
and residual-based transition systems are actively used in providing operational semantics of process calculi
and in demonstrating the consistency of operational and denotational semantics.
The two kinds of transition systems have occasionally been treated alongside each other (see \cite{K96} as an example),
but their general relationship has not been studied for a wide range of existing models.
In a seminal paper, viz. \cite{MR98}, bisimulations between configuration-based and residual-based transition systems have been proved to exist for prime event structures \cite{W89}.
The result of \cite{MR98} has been extended in \cite{BGV17} to more complex event structure models with asymmetric conflict.
The paper \cite{BGV18} demonstrated that when using non-executable events,
the removal operators defined in \cite{MR98,BGV17} to obtain residuals can be tightened in such a way that isomorphisms, rather than just bisimulations,
between the two types of transition systems belonging to a single event structure can be obtained, for a full spectrum of semantics (interleaving, step, pomset, multiset).

Reversible event structures extend event structures to represent reversible computational processes, capable of undoing executed actions by allowing configurations to evolve by eliminating events.
In \cite{PU15,UPY18}, Phillips et. al. determined causal and out-of-causal reversible forms of prime, asymmetric and general event structures
and showed the correspondence between their configurations and traditional ones when there are no reversible events.
In \cite{AC17}, Aubert and Cristescu have provided a true concurrent semantics of a reversible extension of CCS, RCCS (without auto-concurrency, auto-conflict, or recursion), in terms of configuration structures.
In \cite{GPY21}, Graversen et. al. have developed a category of reversible bundle event structures with symmetric conflict
and used the causal subcategory to model semantics of another reversible extension of CCS, CCSK.
They also modified CCSK to control reversibility with a rollback primitive,
and gave, by exploiting the capacity for out-of-causal reversibility, semantics of this kind of CCSK in terms of reversible bundle event structures with asymmetric conflict.
Constructions associating causal reversible prime event structures to reversible occurrence nets and vice versa have been proposed within causal reversibility in \cite{MMPPU20},
as well as within out-of-causal reversibility in \cite{MMP21}.

The aim of this paper is to identify two (configuration-based and residual-based) types of transition system semantics for cause-respecting reversible prime event structures
and to understand how these types relate to each other, which can assist in the construction of algebraic calculi to describe and verify reversible concurrent processes.

This paper is structured as follows.
In Section 2, we start with recalling the syntax of prime and reversible prime event structures and their (step) semantics in terms of configurations and traces.
In Section 3, we define a removal operator, which is useful for constructing model residuals, and demonstrate the correctness of the operator.
In Section 4, we develop two types of transition system semantics for cause-respecting reversible prime event structures and establish bisimulation results between the semantics.
In Section 5, we provide some concluding remarks.
The proofs of the propositions presented here can be found at
www.iis.nsk.su/virb/proofs-AFL-2023.

\section{Reversing in Prime Event Structures}
In this section,
we first recall the notion of prime event structures (PESs) \cite{W80} labeled over the set $L=\{a,b,c,\ldots\}$ of actions,
and then formulate the concept of reversible prime event structures (RPESs) \cite{PU15} and consider their (step) semantics and properties.

The behavior of concurrent systems is formally modelled by event structure models where units of the behavior are represented by events.
There are different ways to relate events.
In prime event structures (PESs),
the dependency between events, called causality, is given by a partial order, and the incompatibility is determined by a conflict relation.
Two events which are neither in causal dependency nor in conflict are considered independent (concurrent).

\begin{definition}
A (labeled) {\em prime event structure (PES)} (over the set $L$ of actions) is a tuple $\mathcal{E}=(E$, $<$, $\sharp$, $l$, $C_0)$, where
\begin{itemize}
\item $E$ is a countable set of {\em events};
\item $< \ \subseteq E \times E$ is an irreflexive partial order (the {\em causality relation}) satisfying the {\em principle of finite causes}:
$\forall e\in E\mydot \lfloor e\rfloor = \{e'\in E \mid e'< e\}$ is finite;
\item $\sharp \subseteq E \times E$ is an irreflexive and symmetric relation (the {\em conflict relation}) satisfying the {\em principle of hereditary conflict}:
$\forall e, e', e''\in E\mydot e < e'$ and $e\ \sharp\ e''$ then $e'\ \sharp\ e''$;
\item $l:E \to L$  is a {\em labeling function};
\item $C_0=\emptyset$ is the {\em initial configuration}\footnote{We add the initial configuration as an empty set to the classical PES definition,
but this does not affect the behavior of the structure in any way,
because the PES progresses by moving from one configuration to another and starting from an empty set.}.
\end{itemize}
\end{definition}

So, the PES is a simple event-based model of concurrent and nondeterministic computations
where events labeled over the set $L$ of actions are considered as atomic, indivisible and instantaneous action occurrences,
some of which can only be executed after another (i.e. there is a causal dependency represented by a partial order $\leq$ between the events) and
some of which might not be executed together (i.e. there is a binary conflict $\sharp$ between the events).
In addition, the principle of finite causes and the principle of conflict inheritance are required.

The PES progresses by executing events, thus moving from one state to another, starting from the initial state, which is an empty set.
A state called a configuration is a set of events that have occurred. 
A subset of events $X \subseteq E$ is {\em left-closed under $<$} iff for all $e \in X$ it holds that $\lfloor e\rfloor \subseteq X$;
is {\em conflict-free} iff for all $e, e'\in  X$ it holds that $\neg (e\ \sharp\ e')$, and we denote it with $CF(X)$.
A subset $C\subseteq E$ is a {\em configuration} of $\mathcal{E}$ iff $C$ is finite, left-closed under $<$ and conflict-free.

Reversible prime event structures (RPESs) \cite{PU15,UPY18} are based on a weaker form of PESs because conflict inheritance may not hold when adding reversibility to PESs.
Also, in RPESs, some events are categorised as reversible, and two relations are added: the reverse causality relation and the prevention relation.
The first one is a dependency relation in the backward direction:
to reverse an event in the current configuration there must be other events on which the event reversibly depends.
The second relation, on the contrary, identifies those events whose presence in the current configuration prevents the event being reversed.

\begin{definition}\label{def_RPES}
A (labeled) {\em reversible prime event structure (RPES)} (over $L$) is a tuple $\mathcal{E}=(E$, $<$, $\sharp$,  $l$, $F$, $\prec$, $\rhd$, $C_0)$, where
\begin{itemize}
\item
$E$ is a countable set of {\em events};
\item
$\sharp \subseteq E \times E$ is an irreflexive and symmetric relation (the {\em conflict relation});
\item
$< \subseteq E \times E$ is an irreflexive partial order (the {\em causality relation}) satisfying:
$\lfloor e\rfloor$ is finite and conflict-free, for every $e \in E$;
\item
$l:E \to L$  is a {\em labeling function};
\item
$F \subseteq E$ are {\em reversible events} being denoted by the set $\underline{F} = \{\underline{e}\ \mid\ e \in  F\}$ such that $\underline{F}\cap E=\emptyset$;

\item $\prec \subseteq E \times \underline{F}$ is the {\em reverse causality relation} satisfying:
$a \prec \underline{a}$ and $\{e \in E\ \mid\ e \prec \underline{a}\}$ is finite and conflict-free, for every $a \in F$;

\item $\rhd \subseteq E\times \underline{F}$ is the {\em prevention relation} such that $\rhd \cap\prec=\emptyset$;

\item
$\ll$ is the transitive {\em sustained causation relation}: $a \ll b$ is defined to mean that $a < b$ and if $a \in F$ then $b \rhd \underline{a}$.\
$\sharp$ is {\em hereditary w.r.t. the sustained causation} $\ll$: if $a\ \sharp\ b\ll c$ then $a\ \sharp\ c$;

\item
$C_0\subseteq E$ is the {\em initial configuration} which is finite, left-closed under $<$ and conflict-free.
\end{itemize}
\end{definition}

It is straightforward to check that any PES is also an RPES with $F=\emptyset$ and $C_0=\emptyset$.
Then, any concept defined for RPESs applies to PESs as well.

\begin{example}\label{0.examp}
Consider the structure $\mathcal{E}_0=(E_0$, $<_0$, $\sharp_0$, $l_0$, $F_0$, $\prec_0$, $\rhd_0$, $C^0_0)$,
where $E_0=\{a,b,c,d,e\}$; $<_0 = \{(b, d), (c,e)\}$;  $\sharp_0 = \{(a,b), (b,a), (b,c), (c,b)\}$; $l_0$ is the identical function;
$F_0=\{b,c\}$; $\prec_0 = \{(b, \underline{b}), (c, \underline{c})\}$; $\rhd_0 = \emptyset$; $C^0_0 = \emptyset$.
It is easy to make sure that the components of the structure $\mathcal{E}_0$ meet the requirements of the corresponding items of Definition~\ref{def_RPES}.
In particular, we see that
$\prec_0 = \{(b, \underline{b}), (c, \underline{c})\}$ and $(b, \underline{b}), (c, \underline{c})\not\in\rhd_0$.
Notice that $\sharp_0$ is not hereditary w.r.t. $<_0$
because $a\ \sharp_0\ b<_0 d$ and $\neg(a\ \sharp_0\ d)$, $b\ \sharp_0\ c<_0 e$ and $\neg(b\ \sharp_0\ e)$, $c\ \sharp_0\ b<_0 d$ and $\neg(c\ \sharp_0\ d)$.
From Definition~\ref{def_RPES}, we know that $x$ and $y$ are in the sustained causation relation iff $x$ causes $y$, and $x$ cannot be reversed as long as $y$ is present.
In $\mathcal{E}_0$, the pairs $(b, d)$ and $(c,e)$ are in the causality relation $<_0$, and the prevention relation $\rhd_0$ is empty.
Therefore, the sustained causation relation $\ll_0$ is empty.
It is easy to see that $\sharp_0$ is hereditary w.r.t. $\ll_0$.
So, the structure $\mathcal{E}_0$ is indeed an RPES.
\hfill$\Diamond$
\end{example}

The RPES progresses by executing events and/or by undoing previously executed events, thus moving from one configuration to another.
The act of moving is a computation step.
Reachable configurations are subsets of events which can be reached from the initial configuration by executing computation steps.
A sequence of computation steps is a trace of the RPES.
\begin{definition}\label{def_conf}
Given an RPES $\mathcal{E} = (E, <, \sharp,  l, F, \prec, \rhd, C_0)$, and $C\subseteq E$ such that $CF(C)$,
\begin{itemize}
\item
for $A \subseteq E$ and $B \subseteq F$, we say that $A \cup \underline{B}$ is {\em enabled at $C$} if
\begin{itemize}
\item[a)] $A\cap C=\emptyset$, $B \subseteq C$, $CF(C \cup A)$;
\item[b)] $\forall e \in A$, $\forall e' \in E$ $\colon$ if $e' < e$ then $e' \in (C \setminus B)$;
\item[c)] $\forall e \in B$, $\forall e' \in E$ $\colon$ if $e' \prec \underline{e}$ then $e' \in (C \setminus (B \setminus \{e\}))$;
\item[d)] $\forall e \in B$, $\forall e' \in E$ $\colon$ if $e' \rhd \underline{e}$ then $e' \not\in (C \cup A)$.
\end{itemize}
If $A \cup \underline{B}$ is enabled at $C$ then $C \stackrel{A\cup \underline{B}}{\longrightarrow} C'=(C\setminus B)\cup A$.
We shall write $l(A\cup\underline{B})=M$ iff $M$ is a multiset over the set $L$ of actions, defined as follows:
$M(a)=\mid \{e\in (A \cup B)\mid l(e)=a\}\mid$ for all $a\in L$.
\item
$C$ is a {\em forwards reachable configuration of $\mathcal{E}$} (from $C_0$) iff
for all $i = 1, \ldots, n$ $(n\geq0)$, there exists a finite set $A_i \subseteq E$ such that $C_{i-1} \stackrel{A_i\cup \underline{\emptyset}_i} \longrightarrow  C_{i}$ and $C_n=C$.
\item
$C$ is a {\em (reachable) configuration of $\mathcal{E}$} (from $C_0$) iff
for all $i = 1, \ldots, n$ $(n\geq0)$, there exist finite sets $A_i \subseteq E$ and $B_i \subseteq F$ such that $C_{i-1} \stackrel{A_i\cup \underline{B}_i} \longrightarrow  C_{i}$ and $C_n=C$.
In this case, $t=(A_1\cup\underline{B}_1)\ldots (A_n\cup\underline{B}_n)$ $(n\geq0)$ is a trace of $\mathcal{E}$ and $last(t)=C_n$.
The set of (reachable) configurations of $\mathcal{E}$ is denoted by $Conf(\mathcal{E})$, and
the set of traces of $\mathcal{E}$ --- by $Traces(\mathcal{E})$.
Clearly, any configuration $C\in Conf(\mathcal{E})$ is conflict-free, and
any prefix of any trace $t\in Traces(\mathcal{E})$ belongs to $Traces(\mathcal{E})$.
\item
Two traces $t=(A_1\cup\underline{B}_1)\ldots (A_n\cup\underline{B}_n)$ $(n\geq0)$ and
$t'=(A'_1\cup\underline{B'}_1)\ldots (A'_m\cup\underline{B'}_m)$ $(m\geq0)$ of $\mathcal{E}$ are called to be {\em equivalent w.r.t. $\sim$} (denoted $t \sim t'$) iff $last(t)=last(t')$.
\end{itemize}
\end{definition}

The last two items of Definition~\ref{def_conf} lead to the following auxiliary

\begin{lemma}\label{lem_0}
Given an RPES $\mathcal{E} = (E, <, \sharp,  l, F, \prec, \rhd, C_0)$, it holds:
\begin{itemize}
\item[(i)] $\{last(t)\mid t\in Traces(\mathcal{E})\}=\Conf(\mathcal{E})$;
\item[(ii)] for any $t\in Traces(\mathcal{E})$,
if $t(A \cup \underline{B})\in Traces(\mathcal{E})$ then $last(t)\stackrel{A\cup\underline{B}}{\rightarrow}last(t(A \cup \underline{B}))$;
\item[(iii)] for any $t,t'\in Traces(\mathcal{E})$,
if $last(t)\stackrel{A\cup \underline{B}}{\rightarrow} last(t')$ then $t(A \cup \underline{B})\in Traces(\mathcal{E})$ and $t(A \cup \underline{B})\sim t'$.
\end{itemize}
\end{lemma}

\begin{example}\label{1.examp}
First, recall the RPES
$\mathcal{E}_0=(E_0$, $<_0$, $\sharp_0$, $l_0$, $F_0$, $\prec_0$, $\rhd_0$, $C^0_0)$ (see Example~\ref{0.examp}) with the components:
$E_0=\{a,b,c,d,e\}$; $<_0 = \{(b, d), (c,e)\}$;  $\sharp_0 = \{(a,b), (b,a), (b,c), (c,b)\}$; $l_0$ is the identical function;
$F_0=\{b,c\}$; $\prec_0 = \{(b, \underline{b}), (c, \underline{c})\}$; $\rhd_0 = \emptyset$; $C^0_0 = \emptyset$.
We shall check
if the sequence $t=(\{b\} \cup \emptyset)(\{d\} \cup \emptyset)(\emptyset \cup \{\underline{b}\})(\{c\} \cup \emptyset)(\{e\} \cup \emptyset)(\emptyset \cup \{\underline{c}\})$ is a trace of $\mathcal{E}_0$,
using Definition~\ref{def_conf}.
First, we need to show that $((A_1 =\{b\}) \cup (\underline{B}_1 =\emptyset))$ is enabled at $C^0_0$.
Item a) is true because $(A_1=\{b\})\cap (C^0_0=\emptyset) = \emptyset$, $B_1=\emptyset\subseteq C^0_0$, and $(\emptyset\cup\{b\})$ is conflict-free.
Item b) is correct, since the event $b$ has no causes, i.e. there is no  $e' \in E_0$ such that $e' <_0 b$.
As $B_1=\emptyset$, items c) and d) are met.
Then, we have $C^0_0 \stackrel{\{b\}\cup \emptyset}{\longrightarrow}C^0_1=\{b\}$.
Second, verify if $((A_2 =\{d\}) \cup (\underline{B}_2 =\emptyset))$ is enabled at $C^0_1$.
We see that $(A_2=\{d\}) \cap (C^0_1=\{b\}) =  \emptyset$, $B_2 = \emptyset \subseteq C^0_1$, and $\{b,d\}$ is conflict-free.
Hence, item a) is correct.
Item b) is met because $d$ has the only cause $b$ belonging to $C^0_1 \setminus B_2$.
Due to $B_2=\emptyset$, items c) and d) are true.
So, we get $C^0_1 \stackrel{\{d\}\cup \emptyset}{\longrightarrow} C^0_2=\{b, d\}$.
Third, make sure that  $((A_3=\emptyset) \cup (\underline{B}_3=\{\underline{b}\}))$ is enabled at $C^0_2$.
Items a) is fulfilled thanks to $A_3 = \emptyset$, $B_3 = \{b\} \subseteq C^0_2$, and $C^0_2$ is conflict-free.
Clearly, item b) is true.
Item c) holds because the only reverse cause for the event $b$ is the event itself, which is in $\{b, d\} = (C^0_2 \setminus (B_3 \setminus \{b\}))$.
As $\rhd_0=\emptyset$, item d) is correct.
Hence, we obtain $C^0_2 \stackrel{\emptyset \cup \{\underline{b}\}}{\longrightarrow} C^0_3=\{d\}$.
Fourth, demonstrate that $((A_4 =\{c\}) \cup (\underline{B}_4 =\emptyset))$ is enabled at $C^0_3$.
We see that $(A_4= \{c\}) \cap (C^0_3 = \{d\}) = \emptyset$, $B_4 = \emptyset \subseteq C^0_3$, and $C^0_3 \cup A_4 = \{c,d\}$ is conflict-free.
This means that item a) is correct.
Item b) is met thanks to the fact that $c$ has no causes.
Because of $B_4=\emptyset$, items c) and d) are met.
Therefore, $C^0_3 \stackrel{\{c\}\cup \emptyset}{\longrightarrow} C^0_4=\{c,d\}$ is true.
Fifth, check that $((A_5 =\{e\}) \cup (\underline{B}_5 =\emptyset))$ is enabled at $C^0_4$.
Since $(A_5=\{e\}) \cap (C^0_4=\{c, d\}) =  \emptyset$, $B_5 = \emptyset \subseteq C^0_4$, $(C^0_4 \cup A_5) = \{c, d, e\}$ is conflict-free, item a) is correct.
As $e$ has the only cause $c$ belonging $C^0_4\setminus B_5$, item b) is met.
Due to $B_5=\emptyset$, items c) and d) are true.
Hence, we get $C^0_4 \stackrel{\{e\}\cup \emptyset}{\longrightarrow} C^0_5=\{c, d, e\}$.
Finally, we examine if $((A_6=\emptyset) \cup (\underline{B}_6=\{\underline{c}\}))$ is enabled at $C^0_5$.
Item a) is fulfilled thanks to $A_6 = \emptyset$, $B_6 = \{c\} \subseteq C^0_5$, and $C^0_5$ is conflict-free.
Obviously, item b) is true.
Item c) holds because the only reverse cause for the event $c$ is the event itself, which is in $\{c, d, e\} = (C^0_5 \setminus (B_6 \setminus \{c\}))$.
Because of $\rhd_0=\emptyset$, item d) is correct.
So, we obtain $C^0_5 \stackrel{\emptyset \cup \{\underline{c}\}}{\longrightarrow} C^0_6=\{d,e\}$.
Thus, $t$ is indeed a trace of $\mathcal{E}_0$.

Reasoning analogously, we get the following configurations of $\mathcal{E}_0$:
$\emptyset$, $\{a\}$, $\{b\}$, $\{c\}$, $\{d\}$, $\{e\}$, $\{b,d\}$, $\{c,e\}$, $\{c,d\}$, $\{b,e\}$, $\{d,e\}$, $\{b,d,e\}$, $\{c,d,e\}$.
Since the event $a$ is independent with each of the events $c$, $d$, $e$, we get the additional configurations:
$\{a,c\}$, $\{a,d\}$, $\{a,e\}$, $\{a,c,e\}$, $\{a,c,d\}$, $\{a,d,e\}$, $\{a,c,d,e\}$.
Since the pair $(a,b)$ ($(b,c)$) is in the conflict relation $\sharp_0$, the events $a$ and $b$ ($b$ and $c$) cannot occur together in any configuration.
Therefore, all the configurations of $\mathcal{E}_0$ are listed above.
Some of the maximal traces of $\mathcal{E}_0$ are:
$(t_1 t_2)^* (\{a\} \cup \emptyset) t_2 t_4 t_2$,\hspace*{0.1cm}
$(t_1 t_2)^* t_4 (t_1 t_2)^* (\{a, c\} \cup \emptyset) t_5$,\hspace*{0.1cm}
$(t_1 t_2)^* t_3 (t_1 t_2)^* t_4 (t_1 t_2)^* (\{a, c\} \cup \emptyset) t_5$,\hspace*{0.1cm}
$(t_1 t_2)^* t_3 (t_1 t_2)^* (\{c\} \cup \emptyset) (\{a, e\} \cup \emptyset) t_5$,\hspace*{0.1cm}
where
$t_1=((\{b\} \cup \emptyset)(\emptyset \cup \{\underline{b}\}))^*$,\hspace*{0.1cm}
$t_2=((\{c\} \cup \emptyset)(\emptyset \cup \{\underline{c}\}))^*$, \hspace*{0.1cm} $t_3=(\{b\} \cup \emptyset) (\{d\} \cup \emptyset) (\emptyset \cup \{\underline{b}\})$,\hspace*{0.1cm}
$t_4=(\{c\} \cup \emptyset) (\{e\} \cup \emptyset) (\emptyset \cup \{\underline{c}\})$,\hspace*{0.1cm}
$t_5=(\emptyset \cup \{\underline{c}\})(\{c\} \cup \emptyset)$.
\medskip

Second, consider the structure $\mathcal{E}_1=(E_1$, $<_1$, $\sharp_1$, $l_1$, $F_1$, $\prec_1$, $\rhd_1$, $C^1_0)$,
where $E_1=\{a,b\}$; $<_1 = \{(a, b)\}$;  $\sharp_1 = \emptyset$; $l_1$ is the identical function;
$F_1=\{a\}$; $\prec_1 = \{(a, \underline{a})\}$; $\rhd_1 = \emptyset$; $C^1_0 = \emptyset$.
It is easy to see that $\mathcal{E}_1$ is an RPES.
As the only pair $(a,b)$ is in the causality relation $<_1$, i.e., the event $a$ has no cause and it causes the event $b$, the event $a$ can occur first and only after that $b$ can happen.
Then, we obtain the forward steps: $\emptyset \stackrel{(\{a\} \cup \emptyset)}{\rightarrow}\{a\}\stackrel{(\{b\} \cup \emptyset)}{\rightarrow}\{a,b\}$.
The intended meaning of $a \prec_1 \underline{a}$ is that the event $a$ can be undone if it has occurred in a configuration.
In this regard, the reverse step $\{a\}\stackrel{(\emptyset \cup \{\underline{a}\})}{\rightarrow}\emptyset$ is possible, thanks to $(b,\underline{a})\not\in\prec_1$ and $\rhd_1 = \emptyset$.
Moreover, the event $a$ can be undone in the configuration $\{a,b\}$ even though the event $b$ is present because $(b,\underline{a})\not\in\rhd_1$.
This means that we can move backwards from $\{a, b\}$ to $\{b\}$ by executing the step $(\emptyset \cup \{\underline{a}\})$.
Therefore, the configurations of $\mathcal{E}_1$ are $\emptyset$, $\{a\}$, $\{b\}$, $\{a,b\}$, and
the traces of $\mathcal{E}_1$ are all prefixes of the trace
$((\{a\}\cup\emptyset)(\emptyset\cup\{\underline{a}\}))^*(\{a\}\cup\emptyset)(\{b\}\cup\emptyset)((\emptyset\cup\{\underline{a}\})(\{a\}\cup\emptyset))^*(\emptyset\cup\{\underline{a}\})$.
\medskip

Third, examine the structure RPES $\mathcal{E}_2=(E_2$, $<_2$, $\sharp_2$, $l_2$, $F_2$, $\prec_2$, $\rhd_2$, $C^2_0)$,
where $E_2=\{a,b\}$; $<_2 = \emptyset$; $\sharp_2 = \emptyset$; $l_2$ is the identical function;
$F_2=\{a\}$; $\prec_2 = \{(a, \underline{a})\}$; $\rhd_2 = \{(b,\underline{a})\}$; $C^2_0 = \emptyset$.
It is not difficult to check that $\mathcal{E}_2$ is an RPES.
As the causality relation $<_2$ and the conflict relation $\sharp_2$ are empty, the events $a$ and $b$ are independent, and, therefore, they can take place in any order.
This leads to the following forward steps:
$\emptyset \stackrel{(\{a\} \cup \emptyset)}\rightarrow \{a\}  \stackrel{(\{b\} \cup \emptyset)}\rightarrow $ $\{a,b\}$ and
$\emptyset \stackrel{(\{b\} \cup \emptyset)}\rightarrow \{b\}  \stackrel{(\{a\} \cup \emptyset)}\rightarrow $ $\{a,b\}$.
Since $b \rhd_2 \underline{a}$, we conclude that $b$ prevents the undoing of $a$, i.e. $a$ cannot be undone if $b$ is present.
So, we can go back from $\{a\}$ to $\emptyset$ by executing the step $(\emptyset \cup \{\underline{a}\})$ and cannot move backwards from $\{a,b\}$.
The configurations of $\mathcal{E}_2$ are $\emptyset$, $\{a\}$, $\{b\}$, $\{a,b\}$,
and the traces of $\mathcal{E}_2$ are all prefixes of the traces
$((\{a\} \cup \emptyset)(\emptyset \cup \{\underline{a}\}))^*(\{a\} \cup \emptyset) (\{b\} \cup \emptyset)$,\hspace*{0.1cm}
$((\{a\} \cup \emptyset)(\emptyset \cup \{\underline{a}\}))^*(\{a,b\} \cup \emptyset)$,\hspace*{0.1cm}
$((\{a\} \cup \emptyset)(\emptyset \cup \{\underline{a}\}))^*(\{b\} \cup \emptyset) (\{a\} \cup \emptyset)$.
\medskip
\newline
It is not difficult to verify the truth of Lemma~\ref{lem_0} for all the RPESs discussed above.
\hfill$\Diamond$
\end{example}
\medskip

RPESs are able to model such a peculiarity of reversible computation as causal-consistent reversibility
which relates reversibility with causality: an event can be undone provided that all of its effects have been undone.
This allows the system to get back to a past state, which could only be reached by forward computation.
This notion of reversibility is natural in reliable concurrent systems since when an error occurs the system tries to go back to a past consistent state.

\begin{definition}
An RPES $\mathcal{E}=(E, <, \sharp,  l, F, \prec, \rhd, C_0)$ is called
\begin{itemize}
\item
{\em cause-respecting} if for any $e, e' \in E$, if $e < e'$ then $e \ll e'$;
\item
{\em causal} if for any $e \in E$ and $u \in F$ it holds: $e \prec \underline{u}$ iff $e = u$, and $e \rhd \underline{u}$ iff $u < e$.
\end{itemize}
\end{definition}

Informally, in the cause-respecting and causal RPES, causes can be only undone if their effects are not present in the current configuration.
Clearly, if the RPES is causal, then it is cause-respecting as well.

\begin{example}\label{2.examp}
First, recall the RPES $\mathcal{E}_0$ (see Examples~\ref{0.examp} and \ref{1.examp}) with the components:
$E_0=\{a,b,c,d,e\}$; $<_0 = \{(b, d), (c,e)\}$;  $\sharp_0 = \{(a,b), (b,a), (b,c), (c,b)\}$; $l_0$ is the identical function;
$F_0=\{b,c\}$; $\prec_0 = \{(b, \underline{b}), (c, \underline{c})\}$; $\rhd_0 = \emptyset$; $C^0_0 = \emptyset$.
We know from Example~\ref{0.examp} that the sustained causation relation $\ll_0$ is empty, because
the causality relation $<_0$ contains the pairs $(b, d)$ and $(c,e)$ and the prevention relation $\rhd_0$ is empty.
Since $\ll_0 \neq <_0$, we have that this RPES is neither cause-respecting nor causal.

Second, consider the RPES $\mathcal{E}_1$ (see Example~\ref{1.examp}) with the components:
$E_1=\{a,b\}$; $<_1 =  \{(a, b)\}$; $\sharp_1 = \emptyset$; $l_1$ is the identical function; $F_1=\{a\}$; $\prec_1 = \{(a, \underline{a})\}$; $\rhd_1 = \emptyset$; $C^1_0 = \emptyset$.
It is easy to see that $\ll_1=\emptyset$, since $<_1 = \{(a, b)\}$ and $(b,\underline{a})\not\in\rhd_1$.
Then, we obtain $<_1\neq\ll_1$.
So, this RPES is neither cause-respecting nor causal.

Third, examine the RPES $\mathcal{E}_2=(E_2$, $<_2$, $\sharp_2$, $l_2$, $F_2$, $\prec_2$, $\rhd_2$, $C^2_0)$ (see Example~\ref{1.examp}) with the components:
$E_2=\{a,b\}$; $<_2 = \emptyset$; $\sharp_2 = \emptyset$; $l_2$ is the identical function; $F_2=\{a\}$; $\prec_2 = \{(a, \underline{a})\}$; $\rhd_2 = \{(b,\underline{a})\}$; $C^2_0 = \emptyset$.
The RPES is cause-respecting,
because the causality relation $<_2$ is empty, and, hence, for the only reversible event $a$ of $\mathcal{E}_3$, the set of its effects is empty, which implies  $<_2 = \ll_2=\emptyset$.
On the other hand, $\mathcal{E}_2$ is not causal, because
there are the events $a$ and $b$ such that $b \rhd_2 \underline{a}$ and $ a \nless_2 b$.

Fourth, treat the RPES $\mathcal{E}_3=(E_3$, $<_3$, $\sharp_3$, $l_3$, $F_3$, $\prec_3$, $\rhd_3$, $C^3_0)$, where
$E_3=\{a,b,c,d\}$; $<_3 = \{(b,d)$, $(c,d)\}$; $\sharp_3 = \{(a,c)$, $(c,a)$, $(a,d)$, $(d,a)\}$; $l_3$ is the identical function;
$F_3=\{b\}$; $\prec_3 = \{(a, \underline{b}), \ (b, \underline{b})\}$; $\rhd_3 = \{(d, \underline{b})\}$ and $C^3_0 = \{b\}$.
Since for the only reversible event $b$, the set of its effects is equal to $\{d\}$ and $d \rhd_3 \underline{b}$ is true, we conclude that the RPES is cause-respecting,
whereas it is not causal because $(a,\underline{b})\in\ \prec_3$ and $a \neq b$.

Finally, consider the RPES $\mathcal{E}_4=(E_4$, $<_4$, $\sharp_4$, $l_4$, $F_4$, $\prec_4$, $\rhd_4$, $C^4_0)$, where
$E_4=\{a,b,c,d\}$; $<_4 = \{(c,d)\}$; $\sharp_4 = \{(a,c)$, $(c,a)$, $(a,d)$, $(d,a)\}$; $l_4$ is the identical function;
$F_4=\{c, b\}$; $\prec_4 = \{ (b, \underline{b})$, $(c, \underline{c})\}$; $\rhd_4 =\{(d, \underline{c})\}$, $C^4_0 = \{b,c\}$.
The RPES is causal and therefore cause-respecting.
This is because $<_4 = \{(c,d)\}$ and $\rhd_4 = \{(d, \underline{c})\}$, and
the reverse cause for the undoing of the only reversible event is the event itself, since we have $F_4=\{b, c\}$ and $ \prec_4 = \{(b, \underline{b}), (c, \underline{c})\}$.
\hfill$\Diamond$
\end{example}

Any cause-respecting RPES with the empty initial configuration can be presented as a PES.
On the other hand, any PES can be converted into a causal and therefore cause-respecting RPES with the empty initial configuration, once we specify which events are to be reversible.
The following facts are slight modifications of Propositions 3.36 and 3.37 from \cite{PU15}.

\begin{proposition}\label{PU15_1}
\ \\
\vspace*{-0.5cm}
\begin{itemize}
\item[(i)]
If $\mathcal{E}=(E, <, \sharp,  l, F, \prec, \rhd, \emptyset)$ is a cause-respecting RPES then $\phi(\mathcal{E})=(E, <, \sharp,  l, \emptyset)$ is a PES.
\item[(ii)]
If $\mathcal{E}=(E,<,\sharp,  l, \emptyset)$ is a PES and $F\subseteq E$ then
$\varphi(\mathcal{E}, F)=(E, <, \sharp,  l, F, \prec, \rhd, \emptyset)$ is a causal RPES,
where $e\prec\underline{e}$ for any $e\in F$, and $e\rhd\underline{e'}$ for any $e \in E$ and $e' \in F$ such that $e'<e$.
Moreover, $\phi(\varphi(\mathcal{E}, F))=\mathcal{E}$.
\end{itemize}
\end{proposition}

The following lemma states specific features of the configurations of the cause-respecting RPES, which are left-closed w.r.t. causality and forwards reachable.
Thanks to Definitions \ref{def_RPES} and \ref{def_conf}, the truth of item (i) follows from Proposition 3.38(1) \cite{PU15}, and the truth of item (ii) --- from Proposition 3.40(2) \cite{PU15}.
\begin{lemma}\label{PU15}
Given a cause-respecting $\mathcal{E}$ and its configuration $C\in Conf(\mathcal{E})$, it holds:
\begin{itemize}
\item[(i)]
$C$ is left-closed under $<$;
\item[(ii)]
if $C$ is reachable, then $C$ is forwards reachable.
\end{itemize}
\end{lemma}

The below example explains the above lemma.

\begin{example}\label{3.examp}
Recall the non-cause-respecting RPES $\mathcal{E}_0$ (with $<_0 = \{(b, d), (c,e)\}$) from Examples~\ref{0.examp}--\ref{2.examp}.
We know that
$\{d\}$, $\{e\}$, $\{b,e\}$, $\{c,d\}$, $\{d,e\}$, $\{b,d,e\}$, $\{c,d,e\}$, $\{a,d\}$, $\{a,e\}$, $\{a,c,d\}$, $\{a,d,e\}$, $\{a,c,e,d\}$ are configurations of $\mathcal{E}_0$.
Clearly, these configurations are not left-closed under $<_0$.
Also, we can reach the configurations only by using a combination of forward and reverse steps,
i.e. the configurations are reachable but not forwards reachable.

Consider the non-cause-respecting RPES $\mathcal{E}_1$ (with $<_1 =  \{(a, b)\}$) from Examples~\ref{1.examp}--\ref{2.examp}.
The configurations of $\mathcal{E}_1$ are $\emptyset$, $\{a\}$, $\{b\}$, $\{a,b\}$.
We see that the configuration $\{b\}$ is not left-closed under $<_1$.
In addition, the configuration $\{b\}$ can only be reached with a combination of forward and reverse steps, but this is not possible when doing only forward steps.

It is easy to check that in the cause-respecting RPES $\mathcal{E}_2$ from Examples~\ref{1.examp}--\ref{2.examp},
all its configurations are left-closed under its causality relation and, moreover, forwards reachable.
\hfill$\Diamond$
\end{example}

\section{Residuals}
The removal operator,
the concept of which is based on deleting already executed configurations (traces) and events that conflict with the events presenting in the configurations (traces),
is necessary for residual semantics.

Introduce the definition of the removal operator for RPESs by using their traces.
\begin{definition}\label{def_rem}
For an RPES $\mathcal{E}=(E, <, \sharp,  l, F, \prec, \rhd, C_0)$ and its trace $t=(A_1\cup\underline{B}_1)$ $\ldots$ $(A_n\cup\underline{B}_n)\in Traces(\mathcal{E})$ ($n\geq0$),
the residual $\mathcal{E}\setminus t$ of $\mathcal{E}$ after $t$ under the removal operator $\setminus$ is defined by induction on $0\leq i\leq n$ as follows:
\begin{itemize}
\item
[$i=0$.]
$\mathcal{E}\setminus (t_0=\epsilon)=\mathcal{E}$.
\item[$i>0$.]
$\mathcal{E}\setminus t_i= (E^{i}$, $<^{i}=<^{i-1}\cap\ (E^{i}\times E^{i})$, $\sharp^{i}=\sharp^{i-1}\cap(E^{i}\times E^{i})$, $l^{i} = l^{i-1}\mid_{E^{i}}$,
$F^{i}$, $\prec^{i}=\prec^{i-1}\cap\ (E^{i}\times \underline{F}^{i})$, $\rhd^{i}=\rhd^{i-1}\cap\ (E^{i}\times \underline{F}^{i})$, $C^{i}_0)$,
with
\begin{itemize}
\item
$E^{i}=E^{i-1}\setminus(\widetilde{A}_i\cup \sharp^{i-1}(\widetilde{A}_i))$, where \\
$\widetilde{A}_i=(A_{i}\setminus F^{i-1})\ \cup\
(\lfloor(A_{i}\setminus F^{i-1})\rfloor\cap{F^{i-1}}=\{\widetilde{a}\in F^{i-1}\mid\exists a\in A_{i}\setminus F^{i-1}\colon\widetilde{a}<^{i-1}a\})$,\\
$\sharp^{i-1}(\widetilde{A}_i)=\{a\in E^{i-1}\mid \exists \widetilde{a}\in \widetilde{A}_i\ \colon\ a\ \sharp^{i-1}\ \widetilde{a}\}$;
\item
$F^{i} = (F^{i-1}\cap E^{i})\setminus \big(\hat{A}_i \cup  \hat{\hat{A}}_i \big)$, where \\
$\hat{A}_i = \{e\in F^{i-1}\mid \exists a\in\sharp^{i-1}(\widetilde{A}_i)\colon a\prec^{i-1}\underline{e}\}$, \\
 $\hat{\hat{A}}_i = \{e\in F^{i-1}\mid \exists a\in \widetilde{A}_i\colon a\rhd^{i-1}\underline{e}\}$;
\item
$C^{i}_0 = ((C^{i-1}_0\setminus B_{i})\cup A_{i}) \cap E^i$.
\end{itemize}
\end{itemize}
$\mathcal{E}\setminus t=\mathcal{E}\setminus t_n$.
\end{definition}

The intuitive interpretation of the above definition is as follows.
In the process of constructing the residual of the RPES after a trace, all the irreversible events occurred in the current computation step, their reversible causes and conflicting events thereof are removed,
yielding a reduction of all the relations, the labelling function and the initial configuration in the residual.
This is due to the fact that all these removed events will never be able to occur in any subsequent step.
In addition, reversible events become irreversible, whenever at least one of their reverse causes and/or at least one of the events preventing their undoing are eliminated
because the reversible events can never be undone afterwards.
At the same time, the other reversible events presented in the current step are retained, since they can be reversed in next steps.

It should be emphasized that for any trace $t$ of the RPES $\varphi(\mathcal{E},\emptyset)$\footnote{See Proposition~\ref{PU15_1}(ii).}, where $\mathcal{E}$ is a PES,
the residual $\varphi(\mathcal{E},\emptyset)\setminus t$ coincides with the residual $\mathcal{E}\ \setminus'\ last(t)$, where $\ \setminus'\ $ is the removal operator defined in \cite{MR98}\footnote{In \cite{MR98},
for the PES $\mathcal{E}=(E, <, \sharp,  l)$ and its configuration $C\in Conf(\mathcal{E})$,
the residual $\mathcal{E}\ \setminus'\  C$
is defined as follows:
$\mathcal{E}\ \setminus'\  C= (E'= E\setminus (C\cup\sharp(C))$, $\leq\cap(E'\times E')$, $\sharp\cap(E'\times E')$,  $l\mid_{E'})$,
where $\sharp(C)$ denotes the events conflicting with the events in $C$.}.

We illustrate the application of the above removal operator with
\begin{example}\label{4.examp}
Consider the RPES $\mathcal{E}_2=(E_2$, $<_2$, $\sharp_2$, $l_2$, $F_2$, $\prec_2$, $\rhd_2$, $C^2_0)$ (see Examples~\ref{1.examp}--\ref{3.examp}) with the components:
$E_2=\{a,b\}$; $<_2 = \emptyset$; $\sharp_2 = \emptyset$; $l_2$ is the identical function;
$F_2=\{a\}$; $\prec_2 = \{(a, \underline{a})\}$; $\rhd_2 = \{(b,\underline{a})\}$; $C^2_0 = \emptyset$.
From Example~\ref{1.examp} we know that the traces of $\mathcal{E}_2$ are
$((\{a\} \cup \emptyset)(\emptyset \cup \{\underline{a}\}))^*$,\hspace*{0.1cm}
$((\{a\} \cup \emptyset)(\emptyset \cup \{\underline{a}\}))^*(\{a\} \cup \emptyset)$,\hspace*{0.1cm}
$((\{a\} \cup \emptyset)(\emptyset \cup \{\underline{a}\}))^*(\{a\} \cup \emptyset) (\{b\} \cup \emptyset)$,\hspace*{0.1cm}
$((\{a\} \cup \emptyset)(\emptyset \cup \{\underline{a}\}))^*(\{a,b\} \cup \emptyset)$,\hspace*{0.1cm}
$((\{a\} \cup \emptyset)(\emptyset \cup \{\underline{a}\}))^*(\{b\} \cup \emptyset)$,\hspace*{0.1cm}
$((\{a\} \cup \emptyset)(\emptyset \cup \{\underline{a}\}))^*(\{b\} \cup \emptyset) (\{a\} \cup \emptyset)$.
\medskip

Applying the removal operator to the RPES $\mathcal{E}_2$ and its traces, we obtain the following structures:
\begin{itemize}
\item[--]
$\tilde{\mathcal{E}}_2=\mathcal{E}_2\setminus (A_1=\{a\} \cup \underline{B}_1=\emptyset)=
(\tilde{E}=E_2$, $\tilde{<}=<_2$, $\tilde{\sharp}=\sharp_2$, $\tilde{l}=l_2$,
$\tilde{F}=F_2$, $\tilde{\prec}=\prec_2$, $\tilde{\rhd}=\rhd_2$, $\tilde{C}_0=\{a\})$,
because $(\widetilde{A}_1\cup\sharp_2(\widetilde{A}_1))=\emptyset$, due to $a\in F_2$, and
$\tilde{C}_0=((C^2_0 = \emptyset) \cup (A_1=\{a\}))\cap (\tilde{E}=\{a,b\}) = \{a\}$;
\item[--]
$\hat{\mathcal{E}}_2=\mathcal{E}_2\setminus(A_1=\{a\} \cup \underline{B}_1=\emptyset)(A_2=\emptyset \cup \underline{B}_2=\{\underline{a}\})=\mathcal{E}_2$,
since $(\widetilde{A}_2\cup\dot{\sharp}(\widetilde{A}_2))=\emptyset$, thanks to $a\in\tilde{F}$, and
$((\tilde{C}_0=\{a\})\setminus B_2=\{a\})\cap\hat{E}_2=\emptyset$;
\item[--]
$\breve{\mathcal{E}}_2=\mathcal{E}_2\setminus (A_1=\{a\} \cup \underline{B}_1=\emptyset)(A_2=\{b\} \cup \underline{B}_2=\emptyset) =
(\breve{E}=\{a\}$, $\breve{<}=\emptyset$, $\breve{\sharp}=\emptyset$; $\breve{l}=l_2|_{\{a\}}$;
$\breve{F}=\emptyset$; $\breve{\prec} = \emptyset$; $\breve{\rhd} = \emptyset$, $\breve{C}_0 = \{a\})$,
because $\widetilde{A}_2=\{b\}$, due to $b\in A_2\setminus \tilde{F}$, $a\not\in\breve{F}$, due to $(b,\underline{a})\in\tilde{\rhd}$, and
$\breve{C}_0 = ((\tilde{C}_0=\{a\}) \cup (A_2=\{b\})) \cap (\breve{E}=\{a\}) = \{a\}$;
\item[--]
$\check{\mathcal{E}}_2=\mathcal{E}_2\setminus(A_1=\{a,b\} \cup \underline{B}_1=\emptyset) = \breve{\mathcal{E}}_2$,
since $\widetilde{A}_1=\{b\}$, due to $b\in A_1\setminus F_2$, $a\not\in\check{F}$, due to $(b,\underline{a})\in\rhd_2$, and
$\check{C}_0 = ((C^0_2=\emptyset) \cup (A_1=\{a,b\})) \cap (\check{E}=\{a\}) = \{a\}=\breve{C}_0$;
\item[--]
$\dot{\mathcal{E}}_2=\mathcal{E}_2\setminus(A_1=\{b\} \cup \underline{B}_1=\emptyset)=
(\dot{E}=\{a\}$, $\dot{<}=\emptyset$, $\dot{\sharp}=\emptyset$, $\dot{l}=l_2|_{\{a\}}$,
$\dot{F}=\emptyset$, $\dot{\prec}=\emptyset$, $\dot{\rhd}=\emptyset$, $\dot{C}_0=\emptyset)$,
because $\widetilde{A}_1=\{b\}$, due to $b\in A_1\setminus F_2$, $a\not\in\dot{F}$, due to $(b,\underline{a})\in\rhd_2$, and
$\dot{C}_0=((C^0_2=\emptyset)\cup (A_1=\{b\})) \cap (\dot{E}=\{a\}) = \emptyset$;
\item[--]
$\ddot{\mathcal{E}}_2=\mathcal{E}_2\setminus(A_1=\{b\} \cup \underline{B}_1=\emptyset)(A_2=\{a\}\cup \underline{B}_2=\emptyset) =
(\ddot{E}=\emptyset$, $\ddot{<}=\emptyset$, $\ddot{\sharp}=\emptyset$, $\ddot{l}=\emptyset$,
$\ddot{F}=\emptyset$, $\ddot{\prec}=\emptyset$, $\ddot{\rhd}=\emptyset$, $\ddot{C}_0=\emptyset)$,
since $\widetilde{A}_2=\{a\}$, due to $a\in A_2\setminus\dot{F}$, and
$\ddot{C}_0=((\dot{C}_0=\emptyset)\cup (A_2=\{a\})) \cap (\dot{E}=\emptyset) = \emptyset$.
\end{itemize}

Notice that the removal operator produces the same residuals after the different traces.
For example, it is easy to see that:

$\mathcal{E}_2\setminus(\{a\} \cup \emptyset)(\emptyset \cup \{\underline{a}\})=\mathcal{E}_2\setminus((\{a\} \cup \emptyset)(\emptyset \cup \{\underline{a}\}))^*$,

$\mathcal{E}_2\setminus(\{a\} \cup \emptyset)=\mathcal{E}_2\setminus((\{a\} \cup \emptyset)(\emptyset \cup \{\underline{a}\}))^*(\{a\} \cup \emptyset)$,

$\mathcal{E}_2\setminus(\{a\} \cup \emptyset)(\{b\} \cup \emptyset) =
\mathcal{E}_2\setminus((\{a\} \cup \emptyset)(\emptyset \cup \{\underline{a}\}))^*(\{a\} \cup \emptyset) (\{b\} \cup \emptyset)$,

$\mathcal{E}_2\setminus(\{a,b\}\cup \emptyset) = \mathcal{E}_2\setminus((\{a\} \cup \emptyset)(\emptyset \cup \{\underline{a}\}))^*(\{a,b\} \cup \emptyset)$,

$\mathcal{E}_2\setminus(\{b\} \cup  \emptyset) = \mathcal{E}_2\setminus((\{a\} \cup \emptyset)(\emptyset \cup \{\underline{a}\}))^*(\{b\} \cup \emptyset)$,

$\mathcal{E}_2\setminus(\{b\} \cup  \emptyset)(\{a\}\cup \emptyset) =
\mathcal{E}_2\setminus((\{a\} \cup \emptyset)(\emptyset \cup \{\underline{a}\}))^*(\{b\} \cup \emptyset) (\{a\} \cup \emptyset)$.\hfill$\Diamond$
\end{example}
\medskip

Below are some technical facts specific to the removal operator for RPESs.

\begin{lemma}\label{lem_rem_0}
Given a cause-respecting RPES $\mathcal{E}=(E, <, \sharp,l, F, \prec,  \rhd, C_0)$,
a trace $t = (A_1 \cup \underline{B_1}) \ldots (A_n \cup \underline{B_n})$
($C_0 \stackrel{A_1 \cup \underline{B_1}}{\rightarrow} C_1$ $\ldots$ $C_{n-1} \stackrel{A_{n} \cup \underline{B_{n}}}{\rightarrow} C_n$) $(n\geq0)$ of $\mathcal{E}$,
and $\mathcal{E}\setminus t=(E^n$, $<^n$, $\sharp^n$, $l^n$, $F^n$, $\prec^n$,  $\rhd^n$, $C^n_0)$, it holds:
\begin{itemize}
\item[(i)]
$E^j\subseteq E^{i}$, $F^j\subseteq F^{i}$, $l^j\subseteq l^{i}$, $\nabla^{j}\subseteq \nabla^{i}$ ($\nabla\in\{<,\sharp,\prec,\rhd\}$), for any $0\leq i\leq j\leq n$;
\item[(ii)]
$\mathcal{E}\setminus t_i$ is a cause-respecting RPES, for any $0\leq i\leq n$;
\item[(iii)]
$B_i\subseteq F^{i-1}$, for any $1\leq i\leq n$;
\item[(iv)]
$A_i\subseteq E^{i-1}$, for any $1\leq i\leq n$;
\item[(v)]
$\widetilde{A}_{i}\subseteq C_n$, for any $1\leq i\leq n$;
\item[(vi)]
$C_0^n=C_n\cap E^n$.
\end{itemize}
\end{lemma}

The following two statements demonstrate compositional properties of the residual operator for cause-respecting RPESs.
\begin{proposition}\label{prop1i}
Given a cause-respecting RPES $\mathcal{E}$ with a trace $t\in Traces(\mathcal{E})$ and
its residual $\mathcal{E}' = \mathcal{E}\setminus t$ with a trace $t'\in Traces(\mathcal{E}')$,
it holds that $tt' \in Traces(\mathcal{E})$, and, moreover, $\mathcal{E}\setminus tt' = \mathcal{E}' \setminus t'$.
\end{proposition}

So, it turned out that the concatenation of any trace $t$ of the cause-respecting RPES $\mathcal{E}$ and any trace $t'$ of the residual $\mathcal{E}\setminus t$ is a trace of $\mathcal{E}$,
and, moreover, the residuals $\mathcal{E}\setminus tt'$ and $\mathcal{E}\setminus t\setminus t'$ coincide.

\begin{example}\label{5.examp}
First, consider the non-cause-respecting $\mathcal{E}_0=(E_0$, $<_0$, $\sharp_0$, $l_0$, $F_0$, $\prec_0$, $\rhd_0$, $C^0_0)$ from Examples~\ref{0.examp}--\ref{3.examp},
where
$E_0=\{a,b,c,d,e\}$; $<_0 = \{(b, d), (c,e)\}$;  $\sharp_0 = \{(a,b), (b,a), (b,c), (c,b)\}$; $l_0$ is the identical function;
$F_0=\{b,c\}$; $\prec_0 = \{(b, \underline{b}), (c, \underline{c})\}$; $\rhd_0 = \emptyset$; $C^0_0 = \emptyset$.
As was demonstrated in Example~\ref{1.examp},
the sequences $(\{b\} \cup \emptyset)$, $(\{b\} \cup \emptyset)(\{d\} \cup \emptyset)$ are traces of $\mathcal{E}_0$.
Construct the following residuals of $\mathcal{E}_0$:
\begin{itemize}
\item[--]
$\dot{\mathcal{E}}_0=\mathcal{E}_0\setminus (A_1=\{b\} \cup \underline{B}_1=\emptyset)=
(\dot{E}_0 = E_0$, $\dot{<}_0=<_0$, $\dot{\sharp}_0=\sharp_0$, $\dot{l}_0=l_0$,
$\dot{F}_0=F_0$, $\dot{\prec}_0=\prec_0$, $\dot{\rhd}_0=\rhd_0$, $\dot{C}_0=\{b\})$,
because $(\widetilde{A}_1\cup\sharp_0(\widetilde{A}_1))=\emptyset$, due to $b\in F_0$,
and, moreover, $\dot{C}_0=((C^0_0 = \emptyset) \cup (A_1=\{b\})) \cap (\dot{E}_0=\{a,b,c,d,e\}) = \{b\}$;

\item[--]
$\ddot{\mathcal{E}}_0=\mathcal{E}_0\setminus (A_1=\{b\} \cup \underline{B}_1=\emptyset)(A_2=\{d\} \cup \underline{B}_2=\emptyset) =
\ddot{E}_0 = \{e\}$, $\ddot{<}_0=\emptyset$, $\ddot{\sharp}_0=\emptyset$, $\ddot{l}_0=\dot{l}_0 |_{\{e\}}$, $\ddot{F}_0=\emptyset$, $\ddot{\prec}_0=\emptyset$, $\ddot{\rhd}_0=\emptyset$, $\ddot{C}_0=\emptyset)$,
because $\widetilde{A}_2 = \{b,d\}$ thanks to $d\in A_2\setminus \dot{F}_0$, $(b,d) \in \dot{<}_0$ and $b\in \dot{F}_0$, and $\dot{\sharp}_0(\widetilde{A}_2)=\{a,c\}$, due to $(a,b), (b,c) \in \dot{\sharp}_0$,
and, moreover, $\dot{C}_0=((\dot{C}_0 = \{b\}) \cup (A_2=\{d\})) \cap (\dot{E}_0=\{e\}) = \emptyset$.
\end{itemize}

It is easy to see that $(\{e\} \cup \emptyset)$ is a trace of $\ddot{\mathcal{E}}_0$,
whereas the sequence $(\{b\} \cup \emptyset)(\{d\} \cup \emptyset)(\{e\} \cup\emptyset)$ is not a trace of $\mathcal{E}_0$.

Using Examples~\ref{1.examp}--\ref{4.examp}, it is not difficult to make sure that Proposition~\ref{prop1i} holds for the cause-respecting RPES $\mathcal{E}_2$.
\hfill$\Diamond$
\end{example}

It is stated below that any suffix $t'$ of any trace $tt'$ of the cause-respecting RPES $\mathcal{E}$ is a trace of the residual $\mathcal{E}\setminus t$.

\begin{proposition}\label{prop1ii}
Given a cause-respecting RPES $\mathcal{E}$ with traces $t',t't''\in Traces(\mathcal{E})$, $t'' \in Traces(\mathcal{E}\setminus t')$ holds.
\end{proposition}

\begin{example}\label{6.examp}
Examine the non-cause-respecting RPES $\mathcal{E}_1$ from Examples~\ref{1.examp}--\ref{3.examp},
with the components:
$E_1=\{a,b\}$; $<_1 =  \{(a, b)\}$; $\sharp_1 = \emptyset$; $l_1$ is the identical function;
$F_1=\{a\}$; $\prec_1 = \{(a, \underline{a})\}$; $\rhd_1 = \emptyset$; $C^1_0 = \emptyset$.
We know that $t'=(\{a\} \cup \emptyset) (\{b\} \cup \emptyset)$ and
$t=(\{a\} \cup \emptyset)(\{b\} \cup \emptyset) (\emptyset \cup \{\underline{a}\})(\{a\} \cup \emptyset)$ are traces of $\mathcal{E}_1$.
Let $t''=(\emptyset \cup \{\underline{a}\})(\{a\} \cup \emptyset)$.
Using Definition~\ref{def_rem}, we obtain the RPES
$\mathcal{E}_1\setminus t' = (E'_1=\emptyset$, $<'_1=\emptyset$, $\sharp'_1=\emptyset$, $l'_1=\emptyset$,
$F'_1=\emptyset$, $\prec'_1 = \emptyset$, $\rhd'_1 = \emptyset$, $C'^1_0 = \emptyset)$.
It is clear that $Traces(\mathcal{E}_1\setminus t') = \emptyset$.
Therefore, we get $t''\not\in Traces(\mathcal{E}_1\setminus t')$.

Using Examples~\ref{1.examp}--\ref{4.examp},
it is not difficult to check that Proposition~\ref{prop1ii} holds for the cause-respecting RPES $\mathcal{E}_2$.
\hfill$\Diamond$
\end{example}

\section{Transition System Semantics for Cause-Respecting RPESs}
\label{assoc.sct}

In this section,
we first give some basic definitions concerning labeled transition systems.
Then, we define the mappings $\TC(\mathcal{E})$ and $\TE(\mathcal{E})$,
which associate two distinct kinds of transition systems
-- one whose states are configurations and one whose states are residuals --
with the RPES $\mathcal{E}$ labeled over the set $L$ of actions.

A transition system $T = (S,\rightarrow,i)$ labeled over a set $\mathcal{L}$ of labels consists of a set of states $S$, a transition relation $\rightarrow\subseteq S\times\mathcal{L}\times S$,
and an initial state $i\in S$.
Two transition systems labeled over $\mathcal{L}$ are {\em isomorphic} if their states can be mapped one-to-one to each other, preserving transitions and initial states.
We call a relation $R\subseteq S\times S'$ a {\em bisimulation} between transition systems $T= (S,\rightarrow,i)$ and $T'= (S',\rightarrow',i')$ over $\mathcal{L}$ iff
$(i,i')\in R$, and for all $(s,s')\in R$ and $l\in \mathcal{L}$:
if $(s,l,s_1)\in\rightarrow$ then $(s',l,s'_1)\in\rightarrow'$ and $(s_1,s'_1)\in R$, for some $s'_1\in S'$;
and if $(s',l,s'_1)\in\rightarrow'$ then $(s,l,s_1)\in\rightarrow$ and $(s_1,s'_1)\in R$, for some $s_1\in S$.
Two transition systems over $\mathcal{L}$ are {\em bisimilar} if there is a bisimulation between them.

For a fixed set $L$ of actions in RPESs,
define the set $\mathbb{L}:= \mathbb{N}_0^L$ (the set of multisets over $L$, or functions from $L$ to the non-negative integers).
The set $\mathbb{L}$ will be used as the set of labels in transition systems.
\medskip

We are ready to define transition systems (labeled over $\mathbb{L}$) with configurations as states.

\begin{definition}\label{TC.def}
For an RPES $\mathcal{E}= (E, <, \sharp,  l, F, \prec, \rhd, C_0)$ over $L$,
\begin{center}
$\TC(\mathcal{E})$ is a transition system $(\Conf(\mathcal{E})$, $\rightharpoondown$, $C_0)$ over $\mathbb{L}$,
\end{center}
where $C\stackrel{M}{\rightharpoondown} C'$
iff
$C\stackrel{(A\cup\underline{B})}{\rightarrow} C'$ in $\mathcal{E}$ and $M=l(A\cup\underline{B})$\footnote{See Definition~\ref{def_conf}.}.
\end{definition}

Let us explain the above definition with

\begin{example}\label{7.examp}
Consider the cause-respecting RPES $\mathcal{E}_2$ from Examples~\ref{1.examp}--\ref{4.examp}.
In Example~\ref{1.examp}, we can see that $C^2_0=\emptyset$ and $\Conf(\mathcal{E}_2)=\{\emptyset$, $\{a\}$, $\{b\}$, $\{a,b\}\}$.
Using Definition~\ref{TC.def}, we obtain
$\rightharpoondown=\{(\emptyset,(\{a\}\cup\emptyset),\{a\})$, $(\{a\},(\emptyset \cup\{a\}),\emptyset)$, $(\{a\},(\{b\}\cup\emptyset),\{a,b\})$,
$(\emptyset,(\{b\}\cup\emptyset),\{b\})$, $(\{b\},(\{a\}\cup\emptyset),\{a,b\})$, $(\emptyset$, $(\{a,b\}\cup\emptyset)$, $\{a,b\})\}$.
A graphical representation of the configuration transition system $\TC(\mathcal{E}_2)$ is shown in Fig.~\ref{CTC2}.
\hfill$\Diamond$
\end{example}

\begin{figure}[htbp]
\begin{center}
\begin{tikzpicture}[line width=0.02cm,>={Latex[length=0.22cm,width=0.1cm]},scale=0.9]
\node[fill=blue!20,at={(1,4)}](E){$\emptyset$};
\node[fill=blue!20,at={(1,6.3)}](E-a){$\{a\}$};
\node[fill=blue!20,at={(4,4)}](E-b){$\{b\}$};
\node[fill=blue!20,at={(4,6.3)}](E-ab){$\{a,b\}$};

\path[-latex]([xshift=-0.cm]E.east)edge[]node[auto,sloped,above,inner sep=0.03cm,pos=0.5,rotate=0]{\small{$(\{a,b\}\cup\emptyset)$}}([xshift=-0.cm]E-ab.south);

\path[-latex]([xshift=-0.13 cm]E.north)edge[]node[auto,sloped,above,inner sep=0.03cm,pos=0.5,rotate=0]{\small{$(\{a\}\cup\emptyset)$}}([xshift=-0.13 cm]E-a.south);
\path[-latex]([xshift=0.13 cm]E-a.south)edge[]node[auto,sloped,below,inner sep=0.03cm,pos=0.5,rotate=180]{\small{$(\emptyset \cup\{a\})$}}([xshift=0.13 cm]E.north);
\path[-latex]([xshift=0. cm]E-a.east)edge[]node[auto,sloped,above,inner sep=0.03cm,pos=0.5,rotate=0]{\small{$(\{b\}\cup\emptyset)$}}([xshift=0.cm]E-ab.west);

\path[-latex]([yshift=-0.1 cm]E.east)edge[]node[auto,sloped,below,inner sep=0.03cm,pos=0.5,rotate=0]{\small{$(\{b\}\cup\emptyset)$}}([yshift=-0.1 cm]E-b.west);
\path[-latex]([yshift=0.1 cm]E-b.north)edge[]node[auto,sloped,below,inner sep=0.03cm,pos=0.5,rotate=0]{\small{$\!\!(\{a\}\cup\emptyset)$}}([yshift=0.cm]E-ab.south);
\end{tikzpicture}
\end{center}
\caption{The configuration transition system $\TC_{}(\mathcal{E}_2)$}\label{CTC2}
\end{figure}
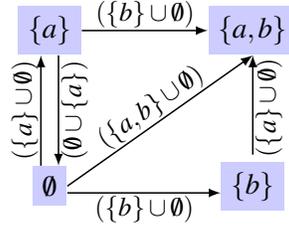

We next propose the definition of labeled transition systems over $\mathbb{L}$ with RPESs as states.
\begin{definition}\label{TR.def}

For an RPES $\mathcal{E}= (E, <, \sharp,  l, F, \prec, \rhd, C_0)$ over $L$,
\begin{center}
$\TE(\mathcal{E})$ is a transition system $(Reach(\mathcal{E})$, $\stackrel{}{\rightharpoonup}$, $\mathcal{E})$ over $\mathbb{L}$,
\end{center}
where $\mathcal{F}\stackrel{M}{\rightharpoonup}\mathcal{F}'$ iff
$\mathcal{F}'=\mathcal{F}\setminus (A\cup \underline{B})$ and $M=l(A \cup \underline{B})$,
and $Reach(\mathcal{E})=\{\mathcal{F} \mid\exists\mathcal{E}_0, \ldots, \mathcal{E}_k$ $(k\geq 0)$ s.t.
$\mathcal{E}_0 = \mathcal{E}\setminus\epsilon$, $\mathcal{E}_k=\mathcal{F}$, and
$\mathcal{E}_{i}\stackrel{l(A \cup \underline{B})}{\rightharpoonup}\mathcal{E}_{i+1}$ $(0\leq i<k)\}$.
\end{definition}

We illustrate the above definition with

\begin{example}\label{8.examp}
Consider the RPES $\mathcal{E}_2$ from Examples~\ref{1.examp}--\ref{4.examp}.
Using Definitions~\ref{def_rem} and \ref{TR.def}, we construct the residual transition system $\TE(\mathcal{E}_2)$ which is depicted in Fig.~\ref{RTC2}.
It is easy to check that the configuration transition system $\TC(\mathcal{E}_2)$ (see Fig.~\ref{CTC2}) and the residual transition system $\TE(\mathcal{E}_2)$ are bisimilar but not isomorphic.
\hfill$\Diamond$
\end{example}

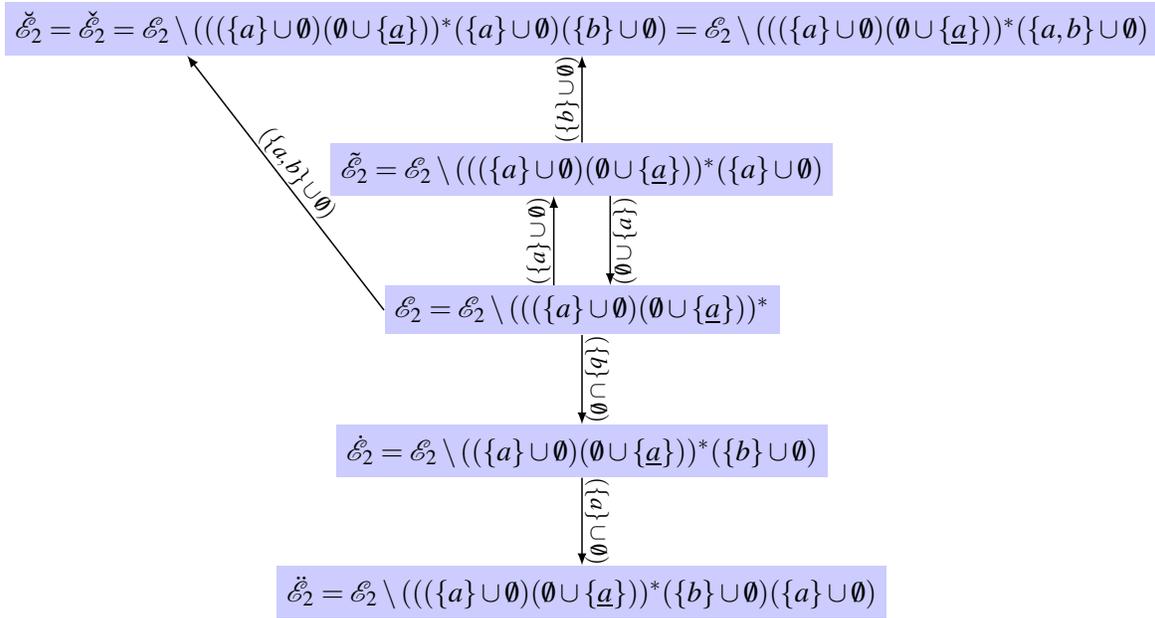
\begin{figure}[htbp]
\begin{center}
\begin{tikzpicture}[line width=0.02cm,>={Latex[length=0.22cm,width=0.1cm]},scale=0.75]
\node[fill=blue!20,at={(0,4)}](E){$\mathcal{E}_2=\mathcal{E}_2\setminus(((\{a\} \cup \emptyset)(\emptyset \cup \{\underline{a}\}))^*$};
\node[fill=blue!20,at={(0,6.5)}](E-a){$\tilde{\mathcal{E}}_2=\mathcal{E}_2\setminus(((\{a\}\cup\emptyset)(\emptyset\cup\{\underline{a}\}))^*(\{a\}\cup\emptyset)$};
\node[fill=blue!20,at={(0,1.5)}](E-b){$\dot{\mathcal{E}}_2=\mathcal{E}_2\setminus((\{a\} \cup \emptyset)(\emptyset \cup \{\underline{a}\}))^*(\{b\} \cup \emptyset)$};
\node[fill=blue!20,at={(0,9)}](E-ab){$\breve{\mathcal{E}}_2=\check{\mathcal{E}}_2=\mathcal{E}_2\setminus(((\{a\} \cup \emptyset)(\emptyset \cup \{\underline{a}\}))^*(\{a\} \cup \emptyset)(\{b\} \cup \emptyset)
=\mathcal{E}_2\setminus(((\{a\} \cup \emptyset)(\emptyset \cup \{\underline{a}\}))^*(\{a,b\} \cup \emptyset)$};
\node[fill=blue!20,at={(0,-1)}](E-ab1){$\ddot{\mathcal{E}}_2=\mathcal{E}_2\setminus(((\{a\} \cup \emptyset)(\emptyset \cup \{\underline{a}\}))^*(\{b\} \cup \emptyset) (\{a\} \cup \emptyset)$};

\path[-latex]([xshift=0.cm]E-a.north)edge[]node[auto,sloped,above,inner sep=0.03cm,pos=0.5,rotate=0]{\footnotesize{$(\{b\}\cup\emptyset)$}}([xshift=0.cm]E-ab.south);

\path[-latex]([xshift=0.5cm]E-a.south)edge[]node[auto,sloped,below,inner sep=0.03cm,pos=0.5,rotate=180]{\footnotesize{$(\emptyset\cup\{a\})$}}([xshift=0.5cm]E.north);
\path[-latex]([xshift=-0.5cm]E.north)edge[]node[auto,sloped,above,inner sep=0.03cm,pos=0.5,rotate=0]{\footnotesize{$(\{a\}\cup\emptyset)$}}([xshift=-0.5cm]E-a.south);

\path[-latex]([xshift=0.cm]E.west)edge[]node[auto,sloped,above,inner sep=0.03cm,pos=0.5,rotate=0]{\footnotesize{$(\{a,b\}\cup\emptyset)$}}([xshift=-7cm]E-ab.south);
\path[-latex]([xshift=0.cm]E.south)edge[]node[auto,sloped,above,inner sep=0.03cm,pos=0.5,rotate=0]{\footnotesize{$(\{b\}\cup\emptyset)$}}([xshift=0.cm]E-b.north);

\path[-latex]([xshift=-0.cm]E-b.south)edge[]node[auto,sloped,above,inner sep=0.03cm,pos=0.5,rotate=0]{\footnotesize{$(\{a\}\cup\emptyset)$}}([xshift=-0.cm]E-ab1.north);

\end{tikzpicture}
\end{center}
\caption{The residual transition system $\TE_{}(\mathcal{E}_2)$}\label{RTC2}
\end{figure}

We establish the relationships between the states and transitions of the configuration-based and residual-based transition systems of the RPES.
\begin{proposition}\label{prop4}
Given a cause-respecting RPES $\mathcal{E}=(E, <, \sharp,l, F, \prec, \rhd, C_0)$ over $L$,
\begin{itemize}
\item[(i)] for any $last(t) \in Conf(\mathcal{E})$, $\mathcal{E}\setminus t\in Reach(\mathcal{E})$;

\item[(ii)] for any $\mathcal{E}'\in Reach(\mathcal{E})$, there is $last(t) \in Conf(\mathcal{E})$ such that $\mathcal{E}'=\mathcal{E}\setminus t$;

\item[(iii)] for any $last(t),\ last(t') \in Conf(\mathcal{E})$,
if $last(t)\stackrel{l(A\cup \underline{B})}{\rightharpoondown} last(t')$
then $\mathcal{E}\setminus t \stackrel{l(A\cup \underline{B})}{\rightharpoonup} \mathcal{E} \setminus t (A\cup \underline{B})$ and $last(t (A\cup \underline{B}))$ $=$ $last(t')$;

\item[(iv)] for  any $\mathcal{E}',\mathcal{E}''\in Reach(\mathcal{E})$,
if $\mathcal{E}'\stackrel{l(A\cup \underline{B})}\rightharpoonup\mathcal{E}''$ then,
for any $last(t)\in Conf(\mathcal{E})$ such that $\mathcal{E}'=\mathcal{E}\setminus t$,
there is $last(t')\in Conf(\mathcal{E})$ such that $\mathcal{E}''=\mathcal{E}\setminus t'$ and $last(t)\stackrel{l(A\cup \underline{B})}{\rightharpoondown} last(t')$.
\end{itemize}
\end{proposition}

\begin{theorem}\label{main.th}
Given a cause-respecting RPES $\mathcal{E}$ over $L$, $\TC(\mathcal{E})$ and $\TR(\mathcal{E})$ are bisimilar and in general not isomorphic.
\end{theorem}
\begin{proof}

From Example~\ref{8.examp} we know that, for the cause-respecting RPES $\mathcal{E}_2$, $\TC(\mathcal{E}_2)$ and $\TR(\mathcal{E}_2)$ are not isomorphic.

We shall check that $\TC(\mathcal{E})$ and $\TR(\mathcal{E})$ are bisimilar for an arbitrary cause-respecting RPES $\mathcal{E}= (E$, $<$, $\sharp$, $L$, $l$, $F$, $\prec$, $\rhd$, $C_0)$.
Due to Lemma~\ref{lem_0}(i) and Propositions~\ref{prop4}(i), we can define a relation $R \subseteq \Conf(\mathcal{E})\times Reach(\mathcal{E})$ as follows:
$R=\{(last(t),\mathcal{E}\setminus t)\mid t\in Traces(\mathcal{E})\}$.

We need to show that $R$ is a bisimulation between $\TC(\mathcal{E})$ and $\TE(\mathcal{E})$.
Clearly, we have that $\epsilon \in Traces(\mathcal{E})$, and, moreover, $C_0 = last(\epsilon)\in\Conf(\mathcal{E})$ and $\mathcal{E} = \mathcal{E}\setminus \epsilon\in Reach(\mathcal{E})$.
So, $(C_0, \mathcal{E})\in R$ holds.
Take an arbitrary $(last(t), \mathcal{E}\setminus t)\in R$.
Suppose that $last(t)\stackrel{l(A\cup\underline{B})}\rightharpoondown C'$ in $\TC(\mathcal{E})$ for some $C' \in Conf(\mathcal{E})$.
By Lemma~\ref{lem_0}(i), there is $t'\in Traces(\mathcal{E})$ such that $C'=last(t')$.
According to Proposition~\ref{prop4}(iii), it is true that
$\mathcal{E}\setminus t \stackrel{l(A\cup \underline{B})}\rightharpoonup \mathcal{E}\setminus t(A\cup \underline{B})$ and $last(t(A\cup \underline{B})) = last(t')$.
Thanks to Lemma~\ref{lem_0}(i), we have $t(A\cup \underline{B})\in Traces(\mathcal{E})$.
Hence, $(C'= last(t(A\cup \underline{B})), \mathcal{E}\setminus t(A\cup \underline{B}))\in R$ holds.
In the opposite direction, assume that $\mathcal{E}\setminus t \stackrel{l(A\cup \underline{B})}\rightharpoonup \mathcal{E}'$ in $\TE(\mathcal{E})$ for some $\mathcal{E}' \in Reach(\mathcal{E})$.
Due to Propositions~\ref{prop4}(iv), for $last(t)\in Conf(\mathcal{E})$,
there is $last(t')\in Conf(\mathcal{E})$ such that $\mathcal{E}' = \mathcal{E}\setminus t'$ and $last(t) \stackrel{l(A\cup \underline{B})}\rightharpoondown last(t')$.
Due to Lemma~\ref{lem_0}(i), $t'\in Traces(\mathcal{E})$ is true.
This implies that $(last(t'), \mathcal{E}\setminus t' = \mathcal{E}')\in R$ holds.
Hence, $R$ is indeed a bisimulation.
\end{proof}

\section{Concluding Remarks}

In this paper, we dealt with two different --  configuration-based and residual-based -- ways of giving (step) transition system semantics for cause-respecting reversible prime event structures
which encompass prime event structures.
For this purpose, we firstly defined (step) semantics from \cite{PU15}, which is based on configurations/traces
obtained by starting with the initial configuration and by executing events and/or undoing previously executed events,
and, secondly, developed a removal operator which is useful for constructing residuals (model fragments) by retaining an appropriate amount of structure during the execution of the model.
We also stated some correctness criteria for the removal operator.
The meaning of the correctness properties is that the obtained residuals do not allow configurations/traces that are disallowed by the original structure.
Also, in some sense, this signifies some compositionality properties of the removal operator.
It turned out that in the context of PESs, the removal operator developed here produces the same residuals as the removal operator proposed in \cite{MR98}.
As our main result, we have obtained a (step) bisimulation between configuration-based and residual-based transition systems of the models under consideration.
The configuration-based method discussed here can be useful in analyzing the state space of reversible concurrent systems whose behavior is represented as RPESs, and
the proposed residual-based method can be suitable for specification and visualization of changes in the structures of reversible concurrent processes during their simulation in tools.
Due to the good compositionality properties of the residual-based transition systems of RPESs and their complementarity and consistency with the configuration-based ones,
it is hoped that the results obtained here may be helpful in demonstrating the correspondence between operational and denotational semantics
of algebraic calculi of reversible concurrent processes, similar to how the results from \cite{BM94,B90,K96} have found their application in traditional (irreversible) process algebras.

As for future work, we plan to broaden the list of studied models by adding flow/bundle/general event structures with symmetric and asymmetric conflict.
Work on extending our approach to out-of-causal reversible prime event structures is under way and has yielded promising intermediate results.
Another future line of our research is to generalize the model of reversible prime event structures with non-executable (impossible) events
(for example, by dropping the transitivity/acyclicity of causality, as well as the principles of finite causes)
in order to obtain isomorphisms between the two types of transition systems of the models, as was done for the corresponding extension of PESs in the paper \cite{BGV18}.
There, the authors have been able to argue that non-executable events are useful in comparative semantics, facilitating the elimination of non-fundamental inconsistencies between models.
Furthermore, isomorphisms between the transition system semantics are expected to allow one
to relate those constructed on configurations and those derived from denotational semantics of process calculi in a tight way.

\bibliographystyle{eptcs}
\bibliography{genpapers.bib}

\end{document}